\documentclass[conference]{IEEEtran}
\IEEEoverridecommandlockouts

\usepackage{amsmath,amssymb,amsthm,mathtools}
\usepackage{bm}
\usepackage{booktabs}
\usepackage{algorithm}
\usepackage{algorithmic}
\usepackage{graphicx}
\usepackage{cite}
\usepackage{hyperref}
\usepackage{color}
\usepackage{multirow}

\newtheorem{proposition}{Proposition}

\newtheorem{remark}{Remark}

\DeclareMathOperator{\KL}{KL}
\DeclareMathOperator{\SW}{SW}

\DeclareMathOperator{\tr}{tr}

\title{Improving D-Optimal Sensor Placement for\\Bearing-Only Localization via Maximum-Entropy Reweighting}

\author{\IEEEauthorblockN{\bf Raktim Bhattacharya}
\IEEEauthorblockA{Department of Aerospace Engineering\\
Texas A\&M University\\
College Station, TX, USA\\
raktim@tamu.edu}}

\begin{document}
\maketitle

\begin{abstract}
In this paper, we present a two-layer architecture for bearing-only sensor placement that improves upon classical D-optimal design. The first layer reweights particles by minimizing Kullback-Leibler divergence from the current distribution subject to a distributional accuracy bound, concentrating mass on regions where the posterior is likely to settle, without reference to the sensor model. The second layer performs D-optimal sensor placement with respect to the reweighted Fisher information matrix, steering sensors toward geometrically informative configurations. Because the two layers are structurally decoupled, the reweighting generalizes across sensing modalities while the placement remains specific to bearing geometry. Systematic experiments on multi-source localization at two noise levels show that this reweighting reduces localization error on average, with the benefit growing as the sensor-to-source ratio increases and as measurements become more informative. The improvement is established in the first few iterations of the sequential procedure and persists as the posterior concentrates.
\end{abstract}


\section{Introduction}\label{sec:intro}

Bearing-only source localization arises in sonar, radar, wildlife tracking, and environmental monitoring whenever mobile sensors measure azimuths to targets at unknown positions \cite{BarShalom2001,Bishop2010}. The standard approach to sensor placement maximizes a scalar function of the Fisher information matrix (FIM), most commonly the log-determinant, a criterion known as D-optimality \cite{Ucinski2005,Martinez2006}.

D-optimal placement requires a probability distribution over the target state at which the FIM is evaluated. In a particle-filter implementation, this distribution is represented by weighted particles. After resampling, the weights are approximately uniform, so the FIM averages over the full prior support without distinction. In early iterations, when the prior spans a broad domain, this uniform average assigns equal importance to particles far from the true source and particles near it. The resulting placement is optimal in an expected sense over the prior, but it does not preferentially exploit the geometry of regions where the posterior mass will concentrate. The limitation is compounded in multi-source problems, where the joint state space is high-dimensional and the posterior has complex structure.

The present paper addresses this limitation by applying the maximum-entropy likelihood synthesis framework of \cite{Bhattacharya2026} to bearing-only sensor placement. The central idea is to insert a reweighting step before the placement optimization. This reweighting solves a constrained optimization: among all particle weight vectors satisfying a distributional accuracy bound (here, a sliced Wasserstein constraint), find the one closest to the current weights in the Kullback-Leibler sense. The solution concentrates mass on particles near the accuracy target while remaining as close to the prior distribution as the constraint permits. The D-optimal placement then operates on the reweighted FIM, which reflects the posterior geometry more faithfully.

The construction proceeds in two layers, each with a distinct role. Layer~1 determines the information requirement: it specifies \emph{what} the posterior should look like, without reference to \emph{how} measurements will be acquired. Layer~2 determines the sensor configuration: it solves the D-optimal placement problem using the reweighted particles as input. The two layers are structurally decoupled. The noise level, the number of sensors, and the sensor positions do not enter the Layer~1 optimization. Only Layer~2 depends on the sensing modality. This separation means the same reweighting applies regardless of whether sensors measure bearings, ranges, time-of-arrival, or any other observable.

\subsection{Contributions}

This paper makes three contributions.
\begin{enumerate}
\item It instantiates the maximum-entropy framework of \cite{Bhattacharya2026} for multi-source bearing-only localization in two dimensions, combining MaxEnt particle reweighting (Layer~1) with D-optimal sensor placement (Layer~2) in a sequential architecture.
\item It provides systematic experimental evidence across 64 configurations ($M \in \{3,\ldots,10\}$ sources, $R \in \{3,\ldots,10\}$ sensors) at two noise levels, identifying the sensor-to-source ratio $R/M$ and the measurement noise~$\sigma$ as the two factors governing the improvement magnitude.
\item It characterizes both the success and failure regimes of the approach: MaxEnt reweighting helps most when geometric freedom is available ($R > M$) and measurements are informative ($\sigma$ not too large), and it can degrade performance, primarily when sensors are scarce relative to sources ($R < M$), because the reweighting sacrifices angular diversity.
\end{enumerate}

\subsection{Related Work}

Optimal sensor placement for target localization has a substantial literature in the D-optimal design framework~\cite{Pukelsheim2006,Ucinski2005,Martinez2006}. Bishop~\cite{Bishop2010} derives optimality conditions specific to bearing-only sensor-target geometry, building on the early FIM analysis of Aidala~\cite{Aidala1979}. Information-driven alternatives include mutual information maximization~\cite{Krause2008} and Bayesian optimal experimental design~\cite{Chaloner1995,Ryan2016}. These methods typically require forward simulation through the sensor model for each candidate configuration, a cost that grows with the number of sensors and the dimensionality of the state space.

The maximum-entropy framework used here was introduced in \cite{Bhattacharya2026} for one-dimensional estimation under Wasserstein distance, Maximum Mean Discrepancy (MMD)~\cite{Gretton2012}, and moment constraints. The present paper extends it to planar bearing-only geometry with multiple sources and sensors. The constrained KL-minimization underlying Layer~1 is an instance of Csisz\'{a}r's $I$-projection~\cite{Csiszar1975}, applied here with a transport-based concentration constraint.

\section{Problem Formulation}\label{sec:formulation}

\subsection{Bearing-Only Measurement Model}

Consider $M$ static sources at unknown positions $\bm{p}_1, \ldots, \bm{p}_M \in \mathbb{R}^2$ in a rectangular domain $\mathcal{D} = [x_{\min}, x_{\max}] \times [y_{\min}, y_{\max}]$, observed by $R$ mobile sensors at positions $\bm{s}_1, \ldots, \bm{s}_R \in \mathcal{D}$. Sensor $r$ measures the bearing to source $m$ as~\cite{BarShalom2001,Aidala1979}
\begin{equation}
\beta_{r,m} = \operatorname{atan2}\!\bigl(p_{m,y} - s_{r,y},\; p_{m,x} - s_{r,x}\bigr) + \eta_{r,m},
\label{eq:bearing_model}
\end{equation}
where $\eta_{r,m} \sim \mathcal{N}(0, \sigma^2)$ and $\sigma$ is the bearing noise standard deviation, common to all sensor-source pairs. Throughout we assume known data association: each measurement $\beta_{r,m}$ is labelled with the source~$m$ that produced it. Relaxations of this assumption are discussed in Section~\ref{sec:discussion}.

The Gaussian observation model~\eqref{eq:bearing_model} induces the bearing likelihood
\begin{equation}
p(\beta \mid \bm{x}, \bm{s}_r, \sigma) = \frac{1}{\sqrt{2\pi}\,\sigma}\exp\!\left\{-\frac{[\Delta\theta(\beta,\hat{\beta}(\bm{x},\bm{s}_r))]^2}{2\sigma^2}\right\},
\label{eq:bearing_lik}
\end{equation}
where $\hat{\beta}(\bm{x},\bm{s}_r) = \operatorname{atan2}(x_y - s_{r,y},\, x_x - s_{r,x})$ is the predicted bearing and $\Delta\theta(\cdot,\cdot) \in [-\pi,\pi]$ wraps the angular difference. Given $N$ particles $\{\bm{x}_i^{(m)}\}_{i=1}^N$ representing the current posterior for source~$m$, the corresponding log-likelihood (dropping the constant $-\tfrac{1}{2}\log 2\pi\sigma^2$) reduces to
\begin{equation}
\ell(\bm{x}_i^{(m)}; \beta, \bm{s}_r, \sigma) = -\frac{1}{2\sigma^2}\bigl[\Delta\theta\bigl(\beta, \hat{\beta}(\bm{x}_i^{(m)}, \bm{s}_r)\bigr)\bigr]^2.
\label{eq:bearing_loglik}
\end{equation}

\subsection{Fisher Information for Bearing Sensors}

The Fisher information matrix (FIM) for a parametric likelihood $p(\beta \mid \bm{x})$ is
\begin{equation}
\bm{I}(\bm{x}) \;=\; \mathbb{E}_{\beta \mid \bm{x}}\!\left[\nabla_{\bm{x}}\log p(\beta\mid \bm{x})\,\nabla_{\bm{x}}\log p(\beta\mid \bm{x})^{\!\top}\right],
\label{eq:fim_general}
\end{equation}
and under standard regularity it equals $-\mathbb{E}[\nabla^2_{\bm{x}}\log p]$~\cite{Pukelsheim2006,Ucinski2005}. For the Gaussian bearing model~\eqref{eq:bearing_lik}, the score equals $\nabla_{\bm{x}}\ell = \sigma^{-2}\,\Delta\theta\,\nabla_{\bm{x}}\hat{\beta}(\bm{x},\bm{s}_r)$, and direct differentiation of $\hat{\beta}$ gives
\begin{equation}
\nabla_{\bm{x}}\hat{\beta}(\bm{x},\bm{s}_r) \;=\; \frac{1}{\rho_r^2}\begin{bmatrix}-(x_y - s_{r,y})\\ \;\;\;(x_x - s_{r,x})\end{bmatrix} \;=\; \frac{\bm{u}_r^{\perp}}{\rho_r},
\label{eq:score}
\end{equation}
where $\rho_r = \|\bm{x}-\bm{s}_r\|$ and $\bm{u}_r^{\perp}$ is the unit vector perpendicular to the line of sight from $\bm{s}_r$ to $\bm{x}$. Substituting~\eqref{eq:score} into~\eqref{eq:fim_general} with $\mathbb{E}[(\Delta\theta)^2]=\sigma^2$ yields the per-particle, per-sensor FIM
\begin{equation}
\bm{I}_r(\bm{x}) \;=\; \frac{1}{\sigma^2}\,\frac{\bm{u}_r^{\perp}(\bm{u}_r^{\perp})^{\!\top}}{\rho_r^2}.
\label{eq:fim_unit}
\end{equation}
This form makes the angular content of~\eqref{eq:bearing_fim} explicit: the outer product $\bm{u}_r^{\perp}(\bm{u}_r^{\perp})^{\!\top}$ encodes the perpendicular-to-line-of-sight direction (the only direction along which a noisy bearing is informative), while $1/\rho_r^2$ is the inverse-square range scaling.

Averaging~\eqref{eq:fim_unit} over the particle approximation of the posterior with weights $\{w_i^{(m)}\}$, and writing $\bm{\delta}_{r,i} = \bm{x}_i^{(m)} - \bm{s}_r$, $\rho_{r,i} = \|\bm{\delta}_{r,i}\|$, $(\delta x_{r,i}, \delta y_{r,i})$ for the components of $\bm{\delta}_{r,i}$, gives the standard bearing-FIM expression~\cite{Bishop2010,Aidala1979}
\begin{equation}
\bm{F}_r^{(m)} = \sum_{i=1}^N \frac{w_i^{(m)}}{\sigma^2 \, \rho_{r,i}^4}
\begin{bmatrix}
\delta y_{r,i}^2 & -\delta x_{r,i}\,\delta y_{r,i} \\
-\delta x_{r,i}\,\delta y_{r,i} & \delta x_{r,i}^2
\end{bmatrix}.
\label{eq:bearing_fim}
\end{equation}
The $1/\rho_{r,i}^4$ factor in~\eqref{eq:bearing_fim} is the product of the $1/\rho^2$ in~\eqref{eq:fim_unit} and a second $1/\rho^2$ arising because the unnormalized perpendicular vector $[-(x_y - s_{r,y}),(x_x - s_{r,x})]^{\!\top}$ is used in place of the unit $\bm{u}_r^{\perp}$. The total FIM for source~$m$ aggregates contributions from all sensors, $\bm{F}^{(m)} = \sum_{r=1}^R \bm{F}_r^{(m)}$, and the joint FIM across sources is block-diagonal because each source has independent particles and the bearing model couples sources only through the shared sensor positions. We write $\bm{F} = \operatorname{blkdiag}(\bm{F}^{(1)},\ldots,\bm{F}^{(M)})$ for the resulting $2M\times 2M$ block-diagonal matrix.

The FIM~\eqref{eq:bearing_fim} has two features that shape the role of weights. The per-particle contribution decays inverse-square with distance (trace $1/(\sigma^2\rho_{r,i}^2)$, from~\eqref{eq:fim_unit}), so nearby particles dominate; and the rank-one outer-product structure carries information only perpendicular to the line of sight, making the FIM anisotropic. Shifting weight toward geometrically informative particles therefore changes the effective FIM substantially.

The D-optimal placement objective maximizes the log-determinant of the total FIM summed over all sources~\cite{Pukelsheim2006,Ucinski2005}:
\begin{equation}
\max_{\bm{s}_1,\ldots,\bm{s}_R \in \mathcal{D}} \; \sum_{m=1}^M \log\det \bm{F}^{(m)}(\bm{s}_1,\ldots,\bm{s}_R; \bm{w}^{(m)}).
\label{eq:doptimal}
\end{equation}
The map $\bm{w}^{(m)} \mapsto \bm{F}^{(m)}$ in~\eqref{eq:bearing_fim} is linear (each weight $w_i^{(m)}$ scales the rank-one contribution of particle~$i$), while the outer $\log\det(\cdot)$ in~\eqref{eq:doptimal} is concave but nonlinear. This is the only mechanism through which Layer~1 communicates with Layer~2: substituting the reweighted weights $\bm{w}^{(m),\mathrm{opt}}$ for the prior weights produces a different FIM, which in turn produces a different optimal sensor configuration.

\subsection{Two-Layer Architecture}\label{sec:twolayer}

The procedure operates sequentially over $T$ iterations. At each iteration~$t$:

\textbf{Layer 1 (MaxEnt reweighting).} For each source $m$, solve the constrained KL minimization
\begin{equation}
\min_{\bm{w} \in \Delta^N} D_{\KL}(\bm{w} \| \bm{w}^-) \quad \text{s.t.} \quad \SW_2(\hat{\pi}_{\bm{w}}, \pi_m^\star) \le \varepsilon,
\label{eq:layer1}
\end{equation}
where $\Delta^N = \{\bm{w} \in \mathbb{R}^N_{\ge 0} : \sum_i w_i = 1\}$ is the probability simplex, $\bm{w}^-$ are the current particle weights (the prior weights produced by the particle filter), $\hat{\pi}_{\bm{w}} = \sum_i w_i \delta_{\bm{x}_i^{(m)}}$ is the weighted empirical measure supported on the existing particle locations, $\SW_2$ is the sliced Wasserstein distance, $\pi_m^\star$ is a target distribution encoding the desired accuracy, and $\varepsilon > 0$ is the accuracy budget. The \emph{inputs} to Layer~1 are therefore $(\{\bm{x}_i^{(m)}\}, \bm{w}^-, \pi_m^\star, \varepsilon, L)$, where $L$ is the number of projection directions used to evaluate $\SW_2$ (introduced below). Notably, neither the sensor positions $\{\bm{s}_r\}$, the noise level $\sigma$, nor the sensor count $R$ enters~\eqref{eq:layer1}; this is the structural decoupling that allows the same reweighting to be reused across sensing modalities.

\emph{Sliced Wasserstein distance.} For two probability measures $\mu, \nu$ on $\mathbb{R}^d$ and $L$ unit directions $\{\bm{d}_\ell\}_{\ell=1}^L$ drawn uniformly on the sphere $\mathbb{S}^{d-1}$, the squared sliced Wasserstein distance~\cite{Rabin2011,Bonnotte2013} is
\begin{equation}
\SW_2^2(\mu, \nu) = \frac{1}{L}\sum_{\ell=1}^L W_2^2(\mu_\ell, \nu_\ell),
\label{eq:sw2}
\end{equation}
where $\mu_\ell$ and $\nu_\ell$ are the one-dimensional pushforwards under the projection $\bm{x} \mapsto \langle \bm{x}, \bm{d}_\ell \rangle$. Each one-dimensional $W_2$ admits a closed-form solution via quantile matching, sidestepping the linear program required by full $W_2$ on $\mathbb{R}^d$. We prefer $\SW_2$ over alternative distances for three reasons. (i) Unlike KL or $\chi^2$, which are infinite when supports disagree, $\SW_2$ is well defined for empirical measures and for Dirac targets, both of which arise here. (ii) Unlike MMD, which depends on a kernel choice and is unitless, $\SW_2$ has the same units as the state space (length), making $\varepsilon$ directly interpretable. (iii) Under a Dirac target, $\SW_2^2$ becomes \emph{linear} in $\bm{w}$ (shown next), enabling the closed-form Lagrangian solution.

\emph{Choice of target distribution.} We set $\pi_m^\star = \delta_{\hat{\bm{x}}_m}$, a Dirac mass at the current weighted mean $\hat{\bm{x}}_m = \sum_i w_i^- \bm{x}_i^{(m)}$, computed from the current weights and held fixed during the optimization. The mean is the canonical $W_2$-projection of the empirical posterior onto the set of point masses: it is the unique minimizer of $\bm{y}\mapsto W_2^2(\hat\pi_{\bm{w}^-},\delta_{\bm{y}}) = \sum_i w_i^-\|\bm{x}_i^{(m)}-\bm{y}\|^2$. Two natural alternatives compare as follows. The \emph{median} (or geometric median in $\mathbb{R}^2$) is the analogous $W_1$ projection and resists heavy tails, but particle-filter posteriors are concentrated rather than heavy-tailed once measurements are informative; in our experiments the mean and median differ by less than $\varepsilon$ after the first few iterations. The posterior \emph{mode} is also a point mass and yields the same closed-form tilt with $\hat{\bm{x}}_m$ replaced by the mode location, which can give a better target when the posterior is strongly asymmetric. A fitted Gaussian or GMM target is more expressive for multi-modal posteriors, but the projected $W_2$ then involves quantile matching of the weighted empirical measure and is no longer linear in $\bm{w}$; Layer~1 remains a convex KL minimization (solvable numerically), or one can substitute a linear moment-matching surrogate such as $\mathbb{E}_{\bm{w}}[\|\bm{x}-\hat{\bm{x}}_m\|^2] \le \tau^2$ that preserves the closed form. We retain the Dirac-at-mean target here for simplicity and direct interpretability and revisit asymmetric and multi-modal extensions in Section~\ref{sec:discussion}.

Because the Dirac target projects to a point mass on every direction, the one-dimensional transport cost reduces to a weighted second moment about that point:
\begin{equation}
\SW_2^2(\hat{\pi}_{\bm{w}}, \delta_{\hat{\bm{x}}_m}) = \sum_{i=1}^N w_i\, g_i, \quad g_i = \frac{1}{L}\sum_{\ell=1}^L \langle \bm{x}_i^{(m)} - \hat{\bm{x}}_m, \bm{d}_\ell \rangle^2.
\label{eq:gi}
\end{equation}
The constraint $\SW_2^2 \le \varepsilon^2$ is therefore \emph{linear} in $\bm{w}$, which makes the optimization tractable.

\emph{Derivation of the exponential tilt.} Replacing the squared $\SW_2$ in~\eqref{eq:layer1} with its linear form~\eqref{eq:gi}, the problem becomes
\begin{equation}
\begin{aligned}
\min_{\bm{w}}\;\;& \sum_{i=1}^N w_i \log\tfrac{w_i}{w_i^-} \\
\text{s.t.}\;\;& \sum_i w_i g_i \le \varepsilon^2,\;\; \sum_i w_i = 1,\;\; w_i \ge 0,
\end{aligned}
\label{eq:layer1_reduced}
\end{equation}
which is a convex program~\cite{BoydVandenberghe2004}: the objective is strictly convex in $\bm{w}$ on $\Delta^N$ (since $w_i\log(w_i/w_i^-)$ is strictly convex), and the feasible set is the intersection of a polytope with a halfspace. Slater's condition holds (the prior $\bm{w}^-$ is strictly feasible whenever the constraint is not already tight), so KKT is both necessary and sufficient. Introducing multiplier $\lambda \ge 0$ for the accuracy constraint and $\mu \in \mathbb{R}$ for the normalization, the Lagrangian is
\begin{multline}
\mathcal{L}(\bm{w},\lambda,\mu) = \sum_{i=1}^N w_i \log \frac{w_i}{w_i^-} \\
+ \lambda\!\left(\sum_{i=1}^N w_i\, g_i - \varepsilon^2\right) - \mu\!\left(\sum_{i=1}^N w_i - 1\right).
\label{eq:lagrangian}
\end{multline}
Setting $\partial \mathcal{L}/\partial w_i = \log(w_i/w_i^-) + 1 + \lambda g_i - \mu = 0$ yields $w_i = w_i^- e^{\mu-1}\,e^{-\lambda g_i}$, and the simplex constraint $\sum_i w_i = 1$ determines $e^{\mu-1} = 1/Z(\lambda)$. Substituting back:
\begin{equation}
w_i^{\mathrm{opt}}(\lambda) = \frac{w_i^- \exp(-\lambda\, g_i)}{Z(\lambda)}, \quad Z(\lambda) = \sum_{j=1}^N w_j^- \exp(-\lambda\, g_j).
\label{eq:exptilt}
\end{equation}
Equation~\eqref{eq:exptilt} is the exponential tilt of the prior $\bm{w}^-$ by the cost vector $\bm{g}$, and is the standard solution of a constrained minimum-KL problem with a linear constraint (Csisz\'{a}r's $I$-projection onto a half-space)~\cite{Csiszar1975,BoydVandenberghe2004}. The non-negativity constraints $w_i \ge 0$ are automatically satisfied by the exponential form. The multiplier $\lambda$ is determined by complementary slackness, $\lambda\bigl(\sum_i w_i^{\mathrm{opt}}(\lambda)\, g_i - \varepsilon^2\bigr) = 0$. If the prior already satisfies $\sum_i w_i^- g_i \le \varepsilon^2$, the constraint is inactive, $\lambda = 0$, and Layer~1 returns $\bm{w}^{\mathrm{opt}} = \bm{w}^-$ unchanged. Otherwise the constraint is active, $\lambda > 0$, and $\lambda$ is found by bisection on the monotone function $h(\lambda) = \sum_i w_i^{\mathrm{opt}}(\lambda)\, g_i$: differentiating $\log Z$ gives $h(\lambda) = -d\log Z/d\lambda$ and $dh/d\lambda = -\mathrm{Var}_{\bm{w}^{\mathrm{opt}}(\lambda)}(\bm{g}) \le 0$, so $h$ is non-increasing in $\lambda$ (concentrating more weight on small-$g_i$ particles as $\lambda$ grows), and bisection converges geometrically.

Particles far from $\hat{\bm{x}}_m$ (large $g_i$) are down-weighted exponentially, while particles near $\hat{\bm{x}}_m$ retain most of their prior mass. The budget $\varepsilon$ controls the degree of concentration: smaller $\varepsilon$ produces sharper reweighting.

\textbf{Layer 2 (Sensor placement).} Using the reweighted particles $\{(\bm{x}_i^{(m)}, w_i^{\mathrm{opt}})\}$ from all sources, solve~\eqref{eq:doptimal} via multi-start L-BFGS-B with analytically computed gradients. By the matrix identity $\mathrm{d}\log\det \bm{A} = \tr(\bm{A}^{-1}\,\mathrm{d}\bm{A})$, the gradient of the D-optimal objective with respect to sensor position $\bm{s}_r$ decomposes as
\begin{equation}
\nabla_{\bm{s}_r} \sum_{m=1}^M \log\det \bm{F}^{(m)} = \sum_{m=1}^M \tr\!\left([\bm{F}^{(m)}]^{-1} \nabla_{\bm{s}_r} \bm{F}_r^{(m)}\right),
\label{eq:grad_dopt}
\end{equation}
where only $\bm{F}_r^{(m)}$ (the contribution of sensor $r$) depends on $\bm{s}_r$, so the gradient of $\bm{F}^{(m)} = \sum_{r'} \bm{F}_{r'}^{(m)}$ reduces to $\nabla_{\bm{s}_r} \bm{F}_r^{(m)}$. The derivative of~\eqref{eq:bearing_fim} with respect to $\bm{s}_r$ involves differentiating the $1/\rho_{r,i}^4$ range factor and the outer-product bearing terms, both of which have closed-form expressions. These analytical gradients avoid finite-difference approximation and allow efficient optimization.

The objective~\eqref{eq:doptimal} is non-convex in the sensor positions because of the $1/\rho_{r,i}^4$ range dependence and the outer-product bearing structure, so multiple local optima exist. We address this with $K$ random restarts, each initializing sensor positions uniformly in $\mathcal{D}$, and retain the solution with the largest objective value. All $K$ restarts run independently and can be parallelized.

\textbf{Bayesian update.} Bearing measurements are simulated from the true source positions with noise~$\sigma$. Given measurements $\{\beta_{r,m}\}$ from the new sensor positions, particle weights for source $m$ are updated multiplicatively using the prior weights $\bm{w}^-$ (not the Layer~1 weights $\bm{w}^{\mathrm{opt}}$, which serve only to guide placement):
\begin{equation}
\tilde{w}_i^{(m)} \propto w_i^- \prod_{r=1}^R \exp\!\bigl(\ell(\bm{x}_i^{(m)}; \beta_{r,m}, \bm{s}_r, \sigma)\bigr),
\label{eq:weight_update}
\end{equation}
where $\ell$ is the log-likelihood from~\eqref{eq:bearing_loglik}. The product over sensors follows from the conditional independence of bearing measurements given the source position. Weights are then normalized to sum to one. To prevent particle degeneracy, systematic resampling is triggered when the effective sample size
\begin{equation}
\mathrm{ESS} = \frac{1}{\sum_{i=1}^N (w_i^{(m)})^2}
\label{eq:ess}
\end{equation}
drops below $N/2$. After resampling, all weights are reset to $1/N$.

The D-optimal baseline uses the same Layer~2 and Bayesian update but omits Layer~1, evaluating the D-optimal objective with the current weights $\bm{w}^-$ directly (exactly $1/N$ after resampling, non-uniform between resampling events). Algorithm~\ref{alg:twolayer} summarizes the procedure.

\begin{algorithm}[t]
\caption{Two-Layer Sequential Sensor Placement}
\label{alg:twolayer}
\begin{algorithmic}[1]
\REQUIRE Prior particles $\{\bm{x}_i^{(m)}\}_{i=1}^N$ for $m = 1,\ldots,M$; initial sensors $\{\bm{s}_r\}_{r=1}^R$; budget $\varepsilon$; iterations $T$
\FOR{$t = 1, \ldots, T$}
  \STATE \textbf{Layer 1:} For each source $m$, solve~\eqref{eq:layer1} via bisection on $\lambda$ to obtain $\bm{w}^{(m),\mathrm{opt}}$
  \STATE \textbf{Layer 2:} Solve~\eqref{eq:doptimal} with weights $\bm{w}^{(m),\mathrm{opt}}$ via multi-start L-BFGS-B to obtain $\{\bm{s}_r^{\mathrm{new}}\}$
  \STATE Simulate bearings $\beta_{r,m}$ from true sources at $\{\bm{s}_r^{\mathrm{new}}\}$
  \STATE Update particle weights via~\eqref{eq:weight_update}
  \STATE If $\mathrm{ESS}$~\eqref{eq:ess} $< N/2$, resample (systematic resampling)
  \STATE $\bm{s}_r \leftarrow \bm{s}_r^{\mathrm{new}}$ for all $r$
\ENDFOR
\RETURN Final particle distributions and sensor positions
\end{algorithmic}
\end{algorithm}

\begin{remark}[Structural decoupling]
Layer~1 is independent of the sensor model. Neither the noise level $\sigma$, the number of sensors $R$, nor the sensor positions $\bm{s}_r$ enter the MaxEnt optimization~\eqref{eq:layer1}. Layer~1 determines what the posterior should look like; Layer~2 determines how to achieve it. This modularity means the same Layer~1 solution applies regardless of the sensing modality used in Layer~2.
\end{remark}

\begin{remark}[Specialization of Layer~2]
In the general framework of \cite{Bhattacharya2026}, Layer~2 fits a parametric sensor likelihood to approximate the ideal likelihood ratio $L_i^\star = w_i^{\mathrm{opt}} / w_i^-$ recovered from Layer~1, and the Bayesian update uses this fitted likelihood when real measurements are unavailable.
In the bearing-only setting, the sensor model is known and real measurements are available at each iteration. Layer~2 therefore specializes to D-optimal placement~\eqref{eq:doptimal} with the reweighted particles entering the FIM computation, and the Bayesian update uses actual bearing measurements rather than the synthesized likelihood.
\end{remark}

\subsection{Analytical Properties}\label{sec:theory}

Three properties of the construction clarify the performance-improvement mechanism and bound the deviation of the reweighted FIM from the prior FIM.

\begin{proposition}[Uniqueness and recovery of D-optimal]\label{prop:uniqueness}
Problem~\eqref{eq:layer1_reduced} is strictly convex on $\Delta^N$ and admits a unique minimizer $\bm{w}^{\mathrm{opt}}$ given by the exponential tilt~\eqref{eq:exptilt}. Moreover, as the budget $\varepsilon \to \infty$ (or whenever $\sum_i w_i^- g_i \le \varepsilon^2$), $\lambda = 0$ and $\bm{w}^{\mathrm{opt}} = \bm{w}^-$, so the two-layer procedure reduces exactly to classical D-optimal placement on $\bm{w}^-$.
\end{proposition}
\begin{proof}
Strict convexity of $\sum_i w_i \log(w_i/w_i^-)$ on the open simplex and linearity of the constraints imply uniqueness. When the prior is already feasible, KKT is satisfied at $\lambda = 0$, which gives $w_i^{\mathrm{opt}} = w_i^-$ by~\eqref{eq:exptilt}.
\end{proof}

\begin{proposition}[Monotone concentration]\label{prop:concentration}
The reweighted measure is strictly closer to the target under the $\SW_2$ pseudo-metric induced by the projection set $\{\bm{d}_\ell\}$ than the prior, whenever the constraint is active:
\begin{equation*}
\SW_2^2(\hat\pi_{\bm{w}^{\mathrm{opt}}},\delta_{\hat{\bm{x}}_m}) \;=\; \varepsilon^2 \;<\; \SW_2^2(\hat\pi_{\bm{w}^-},\delta_{\hat{\bm{x}}_m}).
\end{equation*}
\end{proposition}
\begin{proof}
By complementary slackness, an active constraint forces $\sum_i w_i^{\mathrm{opt}} g_i = \varepsilon^2$, and by assumption $\sum_i w_i^- g_i > \varepsilon^2$. The two equalities are exactly the $\SW_2^2$ values in~\eqref{eq:gi}.
\end{proof}

\begin{proposition}[Sensitivity of the D-optimal objective]\label{prop:sensitivity}
Let $\Phi(\bm{w}) = \log\det \bm{F}^{(m)}(\bm{w})$ with $\bm{F}^{(m)}$ given by~\eqref{eq:bearing_fim}. If both $\bm{F}^{(m)}(\bm{w}^-)$ and $\bm{F}^{(m)}(\bm{w}^{\mathrm{opt}})$ have smallest eigenvalue at least $\lambda_{\min}>0$, then
$|\Phi(\bm{w}^{\mathrm{opt}}) - \Phi(\bm{w}^-)| \le (R K_m / \lambda_{\min}\sigma^2)\,\|\bm{w}^{\mathrm{opt}} - \bm{w}^-\|_1$, where $K_m = \max_{i,r} 1/\rho_{r,i}^2$.
\end{proposition}
\begin{proof}
The mean-value form $\log\det\bm{A}-\log\det\bm{B}=\int_0^1\tr((t\bm{A}+(1-t)\bm{B})^{-1}(\bm{A}-\bm{B}))dt$ and the trace--nuclear-norm inequality give $|\log\det\bm{A}-\log\det\bm{B}|\le\lambda_{\min}^{-1}\|\bm{A}-\bm{B}\|_*$. Writing $\bm{F}^{(m)}(\bm{w})=\sum_i w_i\bm{M}_i$ with $\bm{M}_i=\sum_r \bm{F}_{r,i}^{(m)}$ PSD and $\|\bm{M}_i\|_*\le RK_m/\sigma^2$ gives $\|\Delta\bm{F}^{(m)}\|_*\le (RK_m/\sigma^2)\|\bm{w}^{\mathrm{opt}}-\bm{w}^-\|_1$.
\end{proof}

\begin{remark}[Bisection convergence]\label{rem:bisection}
The function $h(\lambda) = \sum_i w_i^{\mathrm{opt}}(\lambda)\,g_i$ equals $-d\log Z/d\lambda$ from~\eqref{eq:exptilt}; differentiating again gives the cumulant identity $h'(\lambda) = -\mathrm{Var}_{\bm{w}^{\mathrm{opt}}(\lambda)}(\bm{g}) \le 0$. Bisection on $h(\lambda)=\varepsilon^2$ therefore converges geometrically.
\end{remark}

Proposition~\ref{prop:uniqueness} shows that the method strictly generalizes D-optimal placement; Proposition~\ref{prop:concentration} confirms that Layer~1 meets its declared accuracy budget; and Proposition~\ref{prop:sensitivity} makes precise the trade-off controlled by $\varepsilon$: large $\varepsilon$ keeps $\|\bm{w}^{\mathrm{opt}}-\bm{w}^-\|_1$ small and the two placements close, small $\varepsilon$ produces aggressive reweighting and a larger placement deviation.

\section{Benefits of MaxEnt Reweighting}\label{sec:mechanism}

The mechanism is best understood by examining the FIM structure~\eqref{eq:bearing_fim}. Under uniform weighting ($w_i = 1/N$), the FIM averages contributions from all particles equally. In early iterations, when particles span a broad prior over $\mathcal{D}$, distant particles contribute information about bearing geometry at locations far from the true source, diluting the FIM's sensitivity to the geometry that actually matters for posterior concentration.

The MaxEnt reweighting shifts mass toward the accuracy target, subject to the budget constraint. Particles near the current best estimate receive higher weight; particles in the tails are exponentially down-weighted. The resulting FIM emphasizes the bearing geometry in the region where the posterior is concentrating, producing sensor placements that resolve the directional ambiguities most relevant to the current state of knowledge.

Two factors govern the magnitude of the benefit.

\emph{Noise level.} At lower noise ($\sigma = 8^\circ$), each bearing measurement carries substantial geometric information about the source location, so a better-chosen placement produces a noticeably tighter posterior update at every step. The MaxEnt placement advantage established when the posterior is still diffuse (the first few iterations) therefore compounds rapidly over the horizon. At higher noise ($\sigma = 15^\circ$), per-measurement information is reduced; the same placement advantage produces a smaller per-step posterior contraction, and the cumulative improvement is correspondingly smaller.

\emph{Sensor-to-source ratio.} When $R$ comfortably exceeds $M$, the placement problem has geometric degrees of freedom that can be allocated to different posterior regions. The reweighting exploits this freedom by steering sensors toward the most uncertain directions of the concentrated posterior. When $R$ is small relative to $M$ (e.g., $R = 3$ with $M \ge 5$), the geometric conditioning is poor and both methods are tightly constrained by the same angular requirements, leaving little room for the reweighting to improve placement.

\section{Experiments}\label{sec:experiments}

\subsection{Setup}

The experiments place $M \in \{3, 4, \ldots, 10\}$ sources at fixed, well-separated positions in $\mathcal{D} = [0, 20]^2$ and deploy $R \in \{3, 4, \ldots, 10\}$ sensors. Each of the 64 configurations runs for $T = 10$ sequential iterations with $N = 1000$ particles per source, sliced Wasserstein budget $\varepsilon = 0.5$, $L = 50$ random projection directions, and $K = 3$ L-BFGS-B restarts in Layer~2. The performance metric is the root-mean-square error (RMSE) of the weighted particle means relative to the true source positions, averaged over all sources and over 20 Monte Carlo seeds.

Two noise levels bracket the regime of interest:
\begin{itemize}
\item $\sigma = 15^\circ$ (moderate noise): posterior uncertainty is large and bearing geometry offers limited leverage.
\item $\sigma = 8^\circ$ (low noise): geometric effects are amplified, making the placement algorithm more consequential.
\end{itemize}

\subsection{Moderate Noise ($\sigma = 15^\circ$)}\label{sec:results15}

At $\sigma = 15^\circ$, the MaxEnt method reduces RMSE over the D-optimal baseline by a mean of 7.9\% across all $(M, R)$ configurations, with individual improvements ranging from $-21\%$ to $+25\%$. Table~\ref{tab:fixed_results} reports selected entries.

Fig.~\ref{fig:combined_15} displays RMSE versus $R$ for selected source counts alongside the full improvement heatmap over all $(M, R)$ pairs. The heatmap confirms that the advantage concentrates in configurations where $R > M$, i.e., where geometric freedom is available. At low $R/M$ ratios, MaxEnt occasionally degrades performance because the reweighting concentrates the FIM on a posterior subset that does not span the angular diversity needed when sensors are scarce.

\begin{figure}[t]
\centering
\includegraphics[width=\columnwidth]{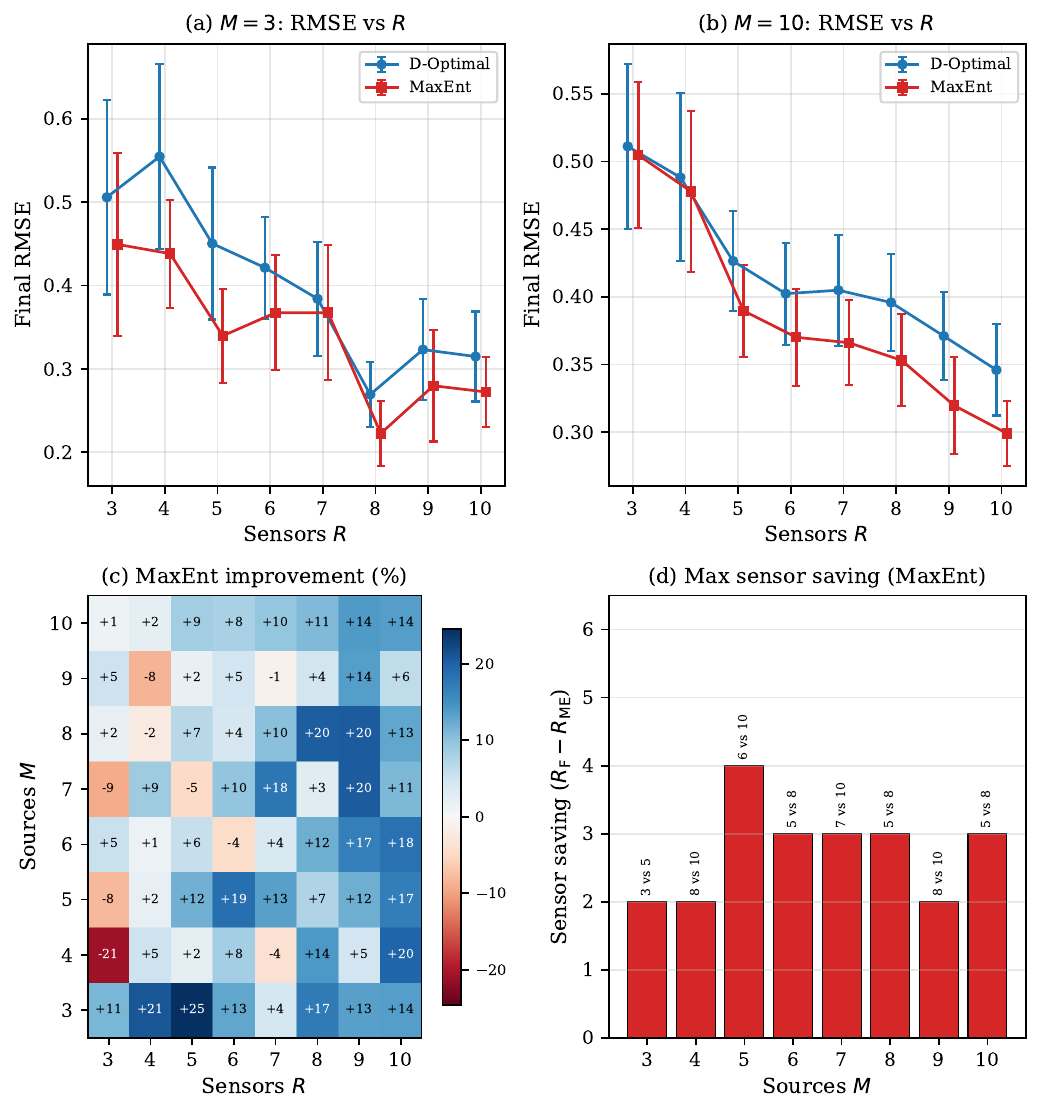}
\caption{Comparison at $\sigma = 15^\circ$. (a),(b)~RMSE versus sensor count for $M = 3$ and $M = 10$. (c)~Improvement heatmap over the full $(M, R)$ grid. (d)~Maximum sensor saving: the largest $R_f - R_m$ such that MaxEnt with $R_m$ sensors achieves RMSE $\le$ D-optimal with $R_f$ sensors.}
\label{fig:combined_15}
\end{figure}

Fig.~\ref{fig:convergence_15} shows the RMSE convergence over the 10 sequential iterations. Both methods improve rapidly during the first 3--4 iterations as the posterior concentrates. The gap between the two methods emerges in the initial iterations and persists thereafter.

\begin{figure}[t]
\centering
\includegraphics[width=\columnwidth]{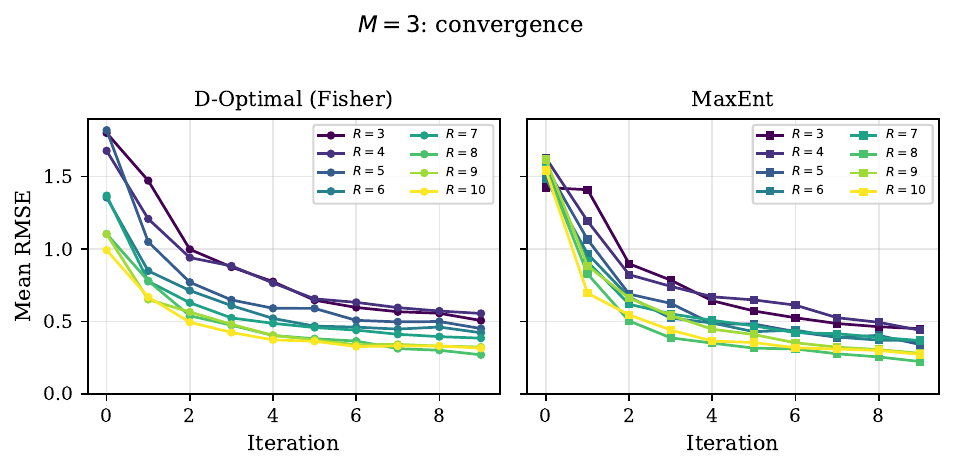}
\caption{RMSE convergence over 10 iterations at $\sigma = 15^\circ$ for selected $(M, R)$ pairs. Shaded bands indicate $\pm 1$ standard deviation over 20 Monte Carlo trials.}
\label{fig:convergence_15}
\end{figure}

\subsection{Low Noise ($\sigma = 8^\circ$)}\label{sec:results8}

At $\sigma = 8^\circ$, the mean improvement rises to 17.2\%, with a range of $-7\%$ to $+34\%$. The lower noise amplifies the MaxEnt advantage because each measurement contracts the posterior more strongly, so the early placement advantage compounds over the horizon.

Fig.~\ref{fig:combined_8} presents the RMSE curves and improvement heatmap at $\sigma = 8^\circ$. The advantage extends more broadly across the $(M, R)$ grid than at $\sigma = 15^\circ$. MaxEnt outperforms D-optimal in nearly all configurations with $R \ge 5$, and the improvement exceeds 20\% across much of the upper-right region of the heatmap.

\begin{figure}[t]
\centering
\includegraphics[width=\columnwidth]{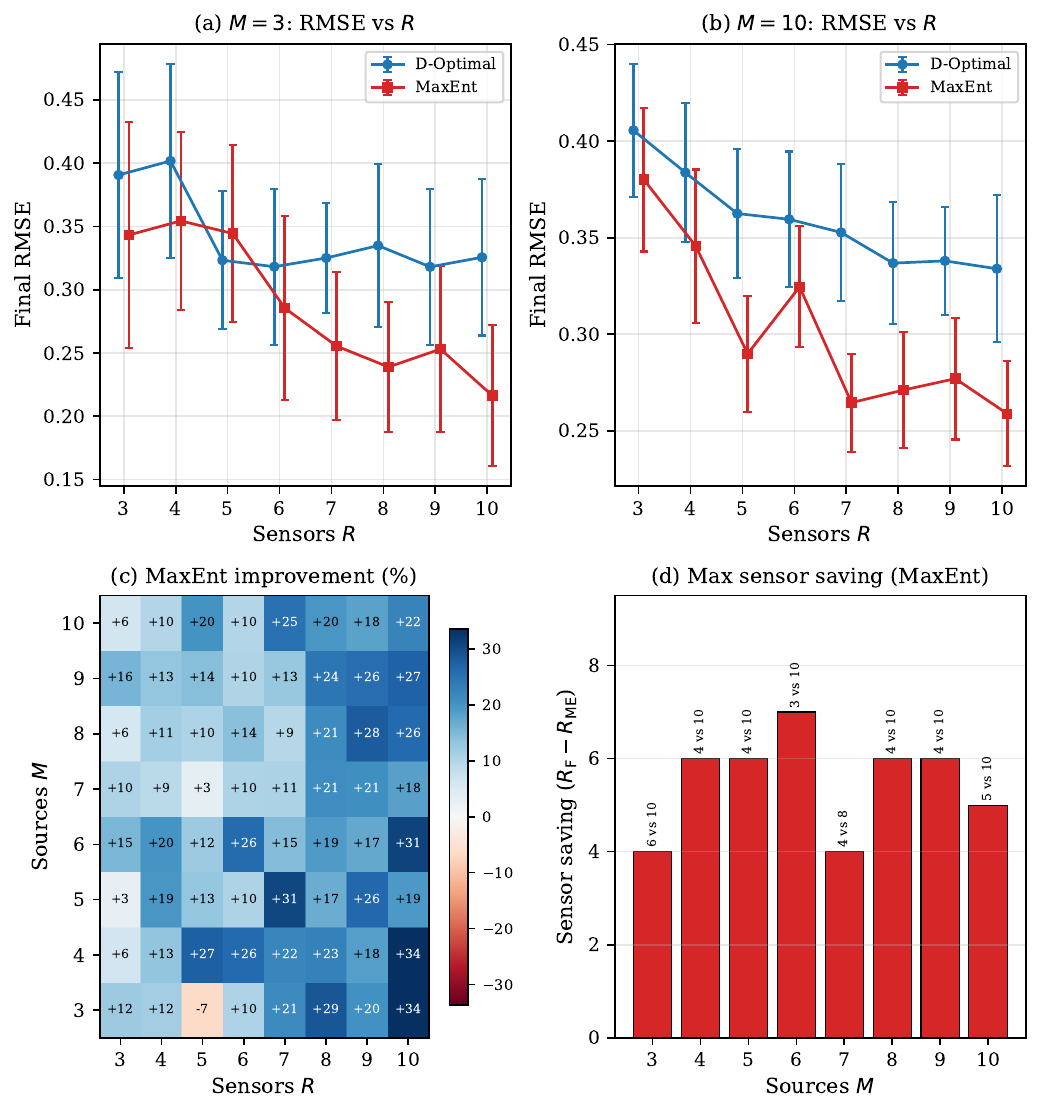}
\caption{Comparison at $\sigma = 8^\circ$. Same layout as Fig.~\ref{fig:combined_15}. The MaxEnt advantage is broader and larger; the sensor saving in panel~(d) reaches 4--7 sensors across all source counts.}
\label{fig:combined_8}
\end{figure}

Fig.~\ref{fig:convergence_8} displays the convergence behavior at $\sigma = 8^\circ$, which mirrors the moderate-noise case but with a larger and more consistent gap.

\begin{figure}[t]
\centering
\includegraphics[width=\columnwidth]{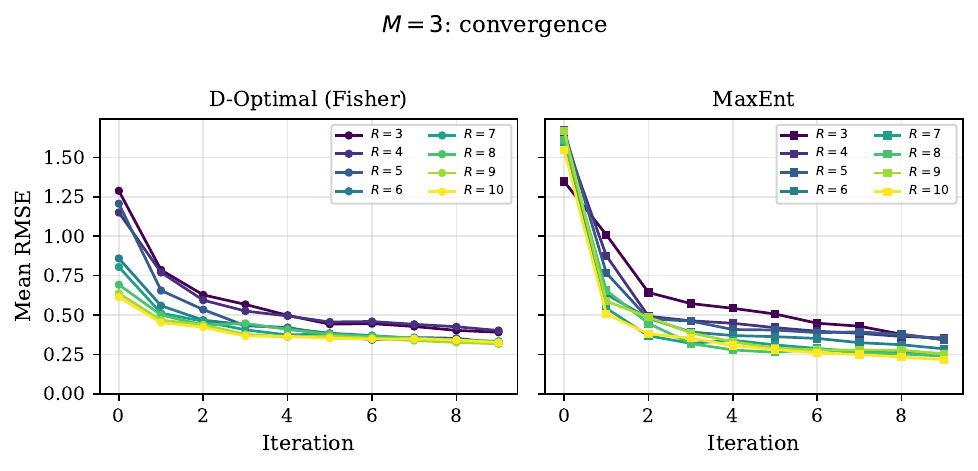}
\caption{RMSE convergence over 10 iterations at $\sigma = 8^\circ$. The MaxEnt advantage is larger and more consistent than at $\sigma = 15^\circ$.}
\label{fig:convergence_8}
\end{figure}

Fig.~\ref{fig:filmstrip} illustrates the spatial evolution of the posterior for a representative configuration ($M = 4$ sources, $R = 10$ sensors) at $\sigma = 8^\circ$. The top row shows D-optimal placement, the bottom row shows MaxEnt. Each column corresponds to a time step. Red circles mark the true source locations; blue circles show the posterior mean estimates, with marker size proportional to the inverse posterior variance (larger markers indicate higher confidence); orange triangles indicate sensor positions. At iteration 0, the prior is uniform over the domain and sensors are initialized near the origin. As measurements accumulate, both methods concentrate the posterior around the true sources, but the trajectories diverge. D-optimal placement distributes sensors to hedge against uncertainty over the full prior support, producing a roughly circular arrangement that provides angular diversity but not targeted refinement. MaxEnt reweighting concentrates mass on particles near the current estimates, steering sensors toward geometrically informative configurations for those regions. By iteration 10, the MaxEnt estimates (bottom right) are visibly closer to the true sources, with final RMSE of 0.06 compared to 0.20 for D-optimal. This example lies in the favorable regime where $R > M$ and the noise is low enough for bearing geometry to discriminate effectively.

\begin{figure}[t]
\centering
\includegraphics[width=\columnwidth]{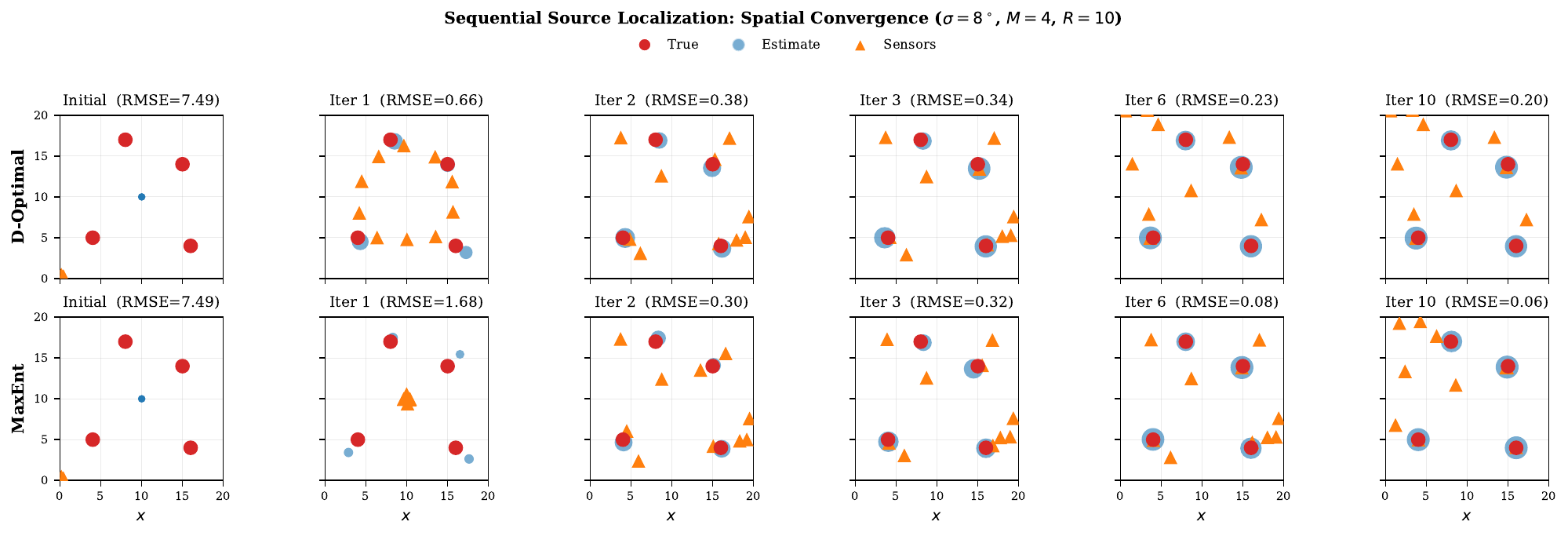}
\caption{Spatial evolution of posterior estimates at $\sigma = 8^\circ$ for $M = 4$ sources and $R = 10$ sensors (single trial). Top row: D-optimal; bottom row: MaxEnt. Red circles: true sources; blue circles: posterior mean estimates (size indicates confidence); orange triangles: sensor positions. In this trial, MaxEnt achieves RMSE $= 0.06$ versus D-optimal's $0.20$.}
\label{fig:filmstrip}
\end{figure}

\begin{table}[t]
\centering
\caption{Final RMSE (averaged over 20 Monte Carlo trials) for selected $(M, R)$ configurations. Percentage improvement of MaxEnt over D-optimal is shown in parentheses.}
\label{tab:fixed_results}
\small
\begin{tabular}{@{}cc cc cc@{}}
\toprule
& & \multicolumn{2}{c}{$\sigma = 15^\circ$} & \multicolumn{2}{c}{$\sigma = 8^\circ$} \\
\cmidrule(lr){3-4} \cmidrule(lr){5-6}
$M$ & $R$ & D-opt & MaxEnt & D-opt & MaxEnt \\
\midrule
3 & 3  & 0.506 & 0.449 {\scriptsize(+11\%)} & 0.391 & 0.343 {\scriptsize(+12\%)} \\
3 & 5  & 0.450 & 0.340 {\scriptsize(+25\%)} & 0.323 & 0.344 {\scriptsize($-$7\%)} \\
3 & 10 & 0.315 & 0.272 {\scriptsize(+14\%)} & 0.326 & 0.216 {\scriptsize(+34\%)} \\
5 & 5  & 0.429 & 0.378 {\scriptsize(+12\%)} & 0.338 & 0.294 {\scriptsize(+13\%)} \\
5 & 10 & 0.339 & 0.282 {\scriptsize(+17\%)} & 0.307 & 0.247 {\scriptsize(+19\%)} \\
8 & 8  & 0.426 & 0.340 {\scriptsize(+20\%)} & 0.346 & 0.273 {\scriptsize(+21\%)} \\
8 & 10 & 0.336 & 0.293 {\scriptsize(+13\%)} & 0.336 & 0.248 {\scriptsize(+26\%)} \\
10& 10 & 0.346 & 0.299 {\scriptsize(+14\%)} & 0.334 & 0.259 {\scriptsize(+22\%)} \\
\bottomrule
\end{tabular}
\end{table}

\subsection{Equivalent Sensor Count}

Panel~(d) of Figs.~\ref{fig:combined_15} and~\ref{fig:combined_8} quantifies the MaxEnt advantage in operational terms as an \emph{equivalent sensor saving}: for each source count~$M$, it reports the largest difference $R_f - R_m$ such that MaxEnt with $R_m$ sensors achieves RMSE no worse than D-optimal with $R_f$ sensors. At $\sigma = 15^\circ$, the saving ranges from 2 to 4 sensors (mean 2.8). At $\sigma = 8^\circ$, it ranges from 4 to 7 (mean 5.5). As a concrete example, at $\sigma = 8^\circ$ with $M = 6$ sources, MaxEnt using $R = 3$ sensors (RMSE $= 0.332$) outperforms D-optimal using $R = 10$ sensors (RMSE $= 0.343$), a saving of 7 sensors. These reductions are operationally significant when each additional sensor incurs hardware and deployment costs.

\section{Discussion}\label{sec:discussion}

The experimental results support a clear characterization of when MaxEnt reweighting helps and when it does not.

\emph{Operating regime.} The benefit is largest when $R > M$ and the posterior retains substantial spatial structure. In early iterations, when uncertainty is high, the reweighting provides the most informative guidance for placement. As the posterior concentrates over successive iterations, both methods converge to similar placements because uniform and reweighted distributions agree in the neighborhood of a point mass. The improvement is thus an \emph{early-iteration} effect that persists in the final RMSE because the posterior trajectory established in the first few iterations determines the quality of the converged estimate.

\emph{Role of measurement noise.} The $1/\sigma^2$ prefactor in the FIM scales every candidate placement equally, leaving $\log\det\bm{F}$ and its gradient $\tr(\bm{F}^{-1}\nabla\bm{F})$ unchanged up to constants; $\sigma$ therefore affects the improvement only through per-measurement information content. At moderate noise ($\sigma = 15^\circ$), measurements are weakly informative and a better placement contracts the posterior only modestly per step (7.9\% mean improvement). At low noise ($\sigma = 8^\circ$), each measurement contracts the posterior more strongly and the MaxEnt placement advantage compounds (17.2\%). At the extremes both methods converge: very high noise extracts little per measurement; very low noise collapses the posterior in a few iterations and both weightings agree.

\emph{Failure modes.} MaxEnt degrades performance in a small fraction of configurations, primarily when $R < M$. In this sensor-starved regime, the reweighting concentrates the FIM on a posterior subset, sacrificing the angular diversity needed for accurate localization. The worst case at $\sigma = 15^\circ$ is a 21\% degradation. When sensors are scarce relative to targets, uniform weighting is the safer choice.

\emph{Beyond Dirac targets and fixed budgets.} The Dirac-at-mean target and fixed budget $\varepsilon = 0.5$ are deliberate simplifications. As Section~\ref{sec:twolayer} notes, non-Dirac targets (Gaussian, fitted GMM, kernel density) lose the closed-form tilt~\eqref{eq:exptilt} and require either a numerical convex solver or a linear moment-matching surrogate. A budget schedule such as $\varepsilon_t = \varepsilon_0/\sqrt{t}$ tightens the reweighting as the posterior concentrates. The largest gains should arise when the posterior is persistently asymmetric or genuinely multi-modal, regimes the present setup does not exercise.

\emph{Alternative information criteria and baselines.} The log-determinant in Layer~2 can be swapped for the trace (A-optimality) or smallest eigenvalue (E-optimality)~\cite{Pukelsheim2006,Ucinski2005} without changing~\eqref{eq:bearing_fim} or $\bm{w}^{\mathrm{opt}}$. Mutual information~\cite{Krause2008} and Bayesian expected information gain~\cite{Chaloner1995,Ryan2016} integrate over the posterior rather than evaluating at one point; they are more expensive but capture nonlinearities the FIM does not. A systematic comparison against these baselines, with the same MaxEnt front-end applied to each, would be a useful next step.

\emph{Computational cost.} The computational overhead of Layer~1 is modest. The FIM computation in Layer~2 scales as $O(MNR)$ per gradient evaluation; with $M = 10$, $N = 1000$, $R = 10$, a single gradient requires approximately 0.5~ms under Numba JIT compilation. Layer~1 adds $O(NL)$ per source ($L = 50$ projections), since the point-mass target eliminates the sorting step of the general sliced Wasserstein computation. This cost is comparable to one FIM evaluation. The full sweep of 2560 configurations ($8 \times 8 \times 2$ noise levels $\times$ 20 seeds) completes in approximately 6 minutes on 8 CPU cores.

\emph{Convergence properties.} The individual components of Algorithm~\ref{alg:twolayer} have well-characterized convergence behavior, but the overall sequential procedure admits only partial guarantees. Layer~1 solves a convex problem: the KL divergence is strictly convex on the simplex and the constraint~\eqref{eq:gi} is linear in~$\bm{w}$, so the exponential tilt~\eqref{eq:exptilt} is the unique global minimizer, and bisection on~$\lambda$ converges at rate $O(2^{-k})$. Layer~2 is non-convex owing to the $1/\rho^4$ range dependence; L-BFGS-B converges to a stationary point, and the $K$-restart strategy reduces the probability of settling on a poor local optimum, but no global optimality guarantee is available. The Bayesian weight update~\eqref{eq:weight_update} is standard importance weighting, for which posterior consistency holds as $N \to \infty$ under regularity conditions satisfied by the bearing model~\cite{BarShalom2001}. The full $T$-iteration procedure is a myopic greedy policy: each iteration optimizes the current placement without lookahead over the remaining horizon. No formal guarantee ensures that this greedy strategy minimizes the $T$-step cumulative RMSE. However, the convergence plots in Figs.~\ref{fig:convergence_15} and~\ref{fig:convergence_8} provide empirical evidence that the procedure stabilizes in practice: both methods reach a stable RMSE plateau within 3--4 iterations, and the gap between MaxEnt and D-optimal persists thereafter without oscillation. A formal analysis of the multi-step regret relative to a non-myopic policy remains an open problem.

\emph{Limitations.} The most significant simplification is the assumption of known data association: each sensor knows which source produced each bearing measurement. Relaxing this assumption introduces a combinatorial assignment problem whose cost grows factorially with the number of sources, requiring integration with methods such as joint probabilistic data association (JPDA)~\cite{BarShalom2009JPDA}, multi-hypothesis tracking (MHT)~\cite{Reid1979}, or probability hypothesis density (PHD) filtering~\cite{Mahler2003}. Real bearing-only deployments are also subject to multipath propagation, clutter, model mismatch, and intermittent or asynchronous communications~\cite{Sinopoli2004}; the present synthetic validation does not exercise these effects. Finally, all experiments are restricted to two-dimensional source positions; the sliced Wasserstein formulation extends to higher dimensions, but the particle count and projection count must both scale to maintain approximation quality.

\section{Conclusion}\label{sec:conclusion}

This paper applied the maximum-entropy likelihood synthesis framework of \cite{Bhattacharya2026} to bearing-only multi-source localization, combining distributional reweighting with D-optimal sensor placement in a two-layer architecture. The structural decoupling between the two layers means the reweighting generalizes to any sensing modality; only the FIM expression in Layer~2 changes when bearings are replaced by range, time-of-arrival, or other measurements. Propositions~\ref{prop:uniqueness}--\ref{prop:sensitivity} establish uniqueness, monotone concentration toward the target, and a quantitative sensitivity bound relating the change in the D-optimal objective to the reweighting magnitude; together with the empirical sweep, they demarcate the regime in which Layer~1 is provably beneficial and the regime in which it provably reduces to classical D-optimal placement.

Beyond addressing the limitations noted above, natural extensions include sequential tracking of moving targets, joint integration with data-association methods (JPDA, MHT, PHD)~\cite{BarShalom2009JPDA,Reid1979,Mahler2003}, comparison against trace, minimum-eigenvalue, mutual-information, and Bayesian-EIG baselines, alternative non-Dirac targets (geometric median, fitted Gaussian or Gaussian mixture) for asymmetric and multi-modal posteriors, adaptive scheduling of the accuracy budget $\varepsilon_t$, and joint optimization of sensor noise precision alongside sensor position.


\end{document}